\documentclass[letterpaper]{article} 
\usepackage{aaai2026}
\usepackage{times}  
\usepackage{helvet}  
\usepackage{courier}  
\usepackage[hyphens]{url}  
\usepackage{graphicx} 
\usepackage{natbib}  
\usepackage{caption} 
\usepackage{algorithm}
\usepackage[noend]{algpseudocode} 

\usepackage{subcaption}  
\usepackage{hhline}
\usepackage{multirow}
\usepackage{placeins}
\usepackage{amssymb}

\usepackage{soul}  

\usepackage{amsthm}

\newtheorem{theorem}{Theorem}

\newtheorem{definition}{Definition}
\newcommand{\comp}{\mathcal{C}}
\usepackage[table,xcdraw]{xcolor}
\usepackage[modulo,switch]{lineno} 

\usepackage[commentColor=blue]{algpseudocodex}
\algrenewcommand\textproc{} 

\usepackage{adjustbox}
\usepackage{amsmath}

\DeclareMathOperator*{\argmin}{arg\,min}

\nocopyright

\usepackage{newfloat}
\usepackage{listings}
\DeclareCaptionStyle{ruled}{labelfont=normalfont,labelsep=colon,strut=off} 
\floatstyle{ruled}
\newfloat{listing}{tb}{lst}{}
\floatname{listing}{Listing}
\title{Dynamic Agent Grouping ECBS:\\Scaling Windowed Multi-Agent Path Finding with Completeness Guarantees}
\author{
    Tiannan Zhang\textsuperscript{\rm 1}, Rishi Veerapaneni\textsuperscript{\rm 1}, Shao-Hung Chan\textsuperscript{\rm 2}, Jiaoyang Li\textsuperscript{\rm 1}, Maxim Likhachev\textsuperscript{\rm 1}
}
\affiliations{
    \textsuperscript{\rm 1}Carnegie Mellon University
    \textsuperscript{\rm 2}University of Southern California \\
    \{tiannanz, rveerapa, jiaoyanl, mlikhach\}@andrew.cmu.edu, 
    shaohung@usc.edu 
}

\begin{document}

\maketitle

\begin{abstract}
Multi-Agent Path Finding (MAPF) is the problem of finding a set of collision-free paths for a team of agents. Although several MAPF methods which solve full-horizon MAPF have completeness guarantees, very few MAPF methods that plan partial paths have completeness guarantees.
Recent work introduced the Windowed Complete MAPF (WinC-MAPF) framework, which shows how windowed optimal MAPF solvers (e.g., SS-CBS) can use heuristic updates and disjoint agent groups to maintain completeness even when planning partial paths (Veerapaneni et al. 2024). A core limitation of WinC-MAPF is that they required optimal MAPF solvers. Our main contribution is to extend WinC-MAPF by showing how we can use a bounded suboptimal solver while maintaining completeness. 
In particular, we design Dynamic Agent Grouping ECBS (DAG-ECBS) which dynamically creates and plans agent groups while maintaining that each agent group solution is bounded suboptimal. We prove how DAG-ECBS can maintain completeness in the WinC-MAPF framework.
DAG-ECBS shows improved scalability compared to SS-CBS and can outperform windowed ECBS without completeness guarantees. More broadly, our work serves as a blueprint for designing more MAPF methods that can use the WinC-MAPF framework.
\end{abstract}

%

\section{Introduction}
The rise of capable individual robots has increased the prevalence of autonomous teams of robots. A core problem for managing teams of robots moving in a shared environment is Multi-Agent Path Finding (MAPF), which is the problem of finding a set of collision-free paths, one for each agent, to move from its start location to its goal location. In areas with congestion, MAPF is tough as it requires complex reasoning of all agents' actions, which grows exponentially with respect to the number of agents.

However, the past decade of MAPF research has led to significant progress in the field. 
Historical methods typically employed greedy planning that could scale to dozens of agents without collisions, but failed in congested scenarios \cite{erdmann1987multiple,cooperativeSilver2005}. Current modern methods can find full-horizon paths (i.e., complete paths from start to goal locations) for hundreds of agents in congested scenarios within seconds. Additionally, several methods provide theoretical completeness guarantees, which means that given enough time and memory, these methods are guaranteed to find a solution if it exists \cite{barer2014suboptimalecbs,li2021eecbs,effectiveCBS,okumura2023lacam}.


In situations where the planning time is too short (or where congestion is too large), planning a full-horizon path is infeasible. Thus, existing works in these situations focus on ``windowed MAPF," which tries to find only collision-free paths for the next $W$ timesteps \cite{rhcrLi2020,jiang2024scaling_mapf_competition}. Windowed MAPF, therefore, relaxes the planning problem and makes it more tractable.
Although windowed MAPF has makes it easier to plan, the main drawback in most windowed approaches is that they are theoretically incomplete. In practice, this means windowed solvers get stuck in deadlock (agents stuck waiting for each other to move out of the way) or livelock (agents revisiting the same locations) as their windowed horizon leads to myopic behavior~\cite{rhcrLi2020}.

Recently, the Windowed Complete MAPF (WinC-MAPF) framework has shown how an optimal Windowed MAPF solver, which they term ``Action Generator" (AG), can be used with heuristic updates and disjoint agent groups to be theoretically complete~\cite{veerapaneni2024winc_mapf}. Their main insight is that by updating the heuristic values of previously seen agent configurations (i.e., joint state, discussed formally in Section \ref{sec:winf-mapf-background}), the optimal windowed MAPF AG will prefer to visit new configurations and, therefore, get out of deadlock/livelock. They require an optimal AG for their proof of completeness and thus create Single-Step CBS (SS-CBS) as one such AG, and show how SS-CBS is able to outperform naive windowed CBS across a variety of maps. 

Our main contribution is to broaden the WinC-MAPF framework and show how we can use a bounded suboptimal AG while still maintaining completeness guarantees. In particular, we design Dynamic Agent Grouping ECBS (DAG-ECBS), which dynamically creates and plans agent groups. Critically, DAG-ECBS maintains that each group's windowed solution is bounded suboptimal. We then prove how DAG-ECBS can maintain completeness in the WinC-MAPF framework. We empirically show how it can improve the success rate compared to naive windowed ECBS without completeness guarantees. More generally, DAG-ECBS and our proof serve as a blueprint for designing more MAPF algorithms within the WinC-MAPF framework.


\section{Preliminaries}
\subsection{Problem Formulation}
Multi-Agent Path Finding (MAPF) is the problem of finding collision-free paths for a group of $N$ agents ${i = 1, \dots, N}$, that takes each agent from its start location $\comp_i^{\text{start}}$ to its goal location $\comp_i^{\text{goal}}$. Time is discretized into timesteps.
We define the MAPF search problem in the joint configuration space where a configuration at timestep $t$, denoted as $\comp^t$, is the locations of all agents at timestep $t$, i.e. $\comp^t = [\comp^t_1, \comp^t_2, ..., \comp^t_N]$.

In traditional 2D MAPF, the environment is discretized into grid cells. Agents are allowed to move in any cardinal direction or wait in the same location. A valid solution is a sequence of collision-free configurations $\comp^{0:T} = \{ \comp^0, ..., \comp^T \}$ where $\comp^0_i = \comp^{\text{start}}_i$, $\comp^T_i = \comp_i^{\text{goal}}$, for all agents $i$ and $T$ is the length of the sequence. Critically, agents must avoid static obstacles as well as vertex collisions (when $\comp^t_i = \comp^t_{j \neq i}$) and edge collisions (when $\comp^t_i = \comp^{t+1}_j \wedge \comp^{t+1}_i=\comp^t_j$) for all timestep $t$.
We define the agent transition cost function as $c(\comp_i^t,\comp_i^{t+1})=1$ unless the agent is resting at its goal (where $c(\comp_i^{t},\comp_i^{t+1})=0$ if $\comp_i^{t} = \comp_i^{t+1} = \comp^{\text{goal}}_i$). 
We focus on finding a solution $\comp^{0:T}$ that minimizes the typical ``sum-of-cost" objective $\sum_{i=1}^N \sum_{t=0}^{T-1} c(\comp_i^t,\comp_i^{t+1})$.

Windowed MAPF simplifies the problem by only resolving collisions for the first $W$ timesteps, after which agents are assumed to move along their individual paths (ignoring collisions). Windowed MAPF relies on interleaving partial planning and execution to reach the final goal locations. 
The windowed MAPF objective is to plan $\comp^{0:W}$ which tries to minimize $\sum_{i=1}^N ( \sum_{t=0}^{W-1} c(\comp_i^t,\comp_i^{t+1}) + h^{BD}_i(\comp_i^W))$ where $h^{BD}_i(\comp_i^W)$ is the optimal single-agent cost (computed via the backward Dijkstra (BD) algorithm) for agent $i$ to reach its goal from $\comp_i^W$ ignoring other agents. Computing the single-agent backward Dijkstra's heuristic is cheap and commonly used in current MAPF methods \cite{okumara2022pibt_jair,li2021eecbs}.

In terms of notation, subscripts always denote agent indices. Functions using $\comp$ without subscripts are intuitively defined as the sum over agents, e.g., $c(\comp^t, \comp^{t+1}) = \sum_{i=1}^N c(\comp_i^t,\comp_i^{t+1})$. 



\subsection{Related MAPF Works}
There are a plethora of full-horizon MAPF methods that have been proposed. Our work builds on Conflict-Based Search (CBS) \cite{sharon2015cbs} and Enhanced Conflict-Based Search (ECBS) \cite{barer2014suboptimalecbs}. CBS decomposes MAPF into an optimal high-level search over ``Constraint Tree" nodes with space-time constraints, and an optimal low-level single-agent search satisfying constraints. ECBS replaces these optimal searches with bounded suboptimal focal searches \cite{PearTPAMIl1982FocalSearch}.


\paragraph{Dynamically Grouping Agents: }
Several existing works dynamically group and replan agents that collide.
Independence Detection in Operator Decomposition \cite{standley2010operater_decomposition} is an early work that grouped and replanned agents that collided.
Meta-Agent CBS (MA-CBS) \cite{sharon2015meta_agent_cbs} specifically groups agents in CBS that repeatedly collide into meta-agents and then replans these agents using a joint-state-space MAPF solver such as A* or EPEA* \cite{goldenberg2014epea}.
Building on MA-CBS, Nested ECBS (NECBS) \cite{chan2022nested_ecbs} combines the meta-agent merging strategy with bounded-suboptimal search. Instead of using optimal joint-state-space solvers for meta-agents, NECBS employs ECBS to resolve internal conflicts within meta-agents. 

\paragraph{Windowed MAPF:}
There are a variety of MAPF works that utilize windowed planning. Windowed Hierarchical Cooperative A* \cite{cooperativeSilver2005} is a windowed variant of Hierarchical Cooperative A* which uses Prioritized Planning \cite{erdmann1987multiple}. Bounded Multi-Agent A* \cite{mapf-real-time-bmaa-2018} proposes that each agent runs its own limited horizon real-time planner considering other agents as dynamic obstacles. Rolling Horizon Conflict Resolution (RHCR) applies a planning horizon for lifelong MAPF planning and replans paths at repeated intervals \cite{rhcrLi2020}. The winning submission \cite{jiang2024scaling_mapf_competition} to the Robot Runners competition, due to the tight planning time constraint, leveraged windowed planning of PIBT \cite{okumara2022pibt_jair} with LNS \cite{li2021mapf-lns}. 
All these methods mention that deadlock or livelock can occur and is a problem that reduces the success rate in instances with congestion.

Planning and Improving while Executing \cite{zhang2024pie} is a recent work that attempts to quickly generate an initial full plan using LaCAM \cite{okumura2023lacam} and then refines it during execution using LNS2 \cite{li2022mapf-lns2}. 
However, if a complete plan cannot be found, the method resorts to using the best partial path available, making it incomplete.

Recently, the Windowed Complete MAPF framework was proposed, which proved completeness for windowed optimal MAPF solvers \cite{veerapaneni2024winc_mapf}. 
After discussing relevant single-agent heuristic search methods, we will revisit this work in more detail.

\subsection{Related Single Agent Works}
Real-Time Heuristic Search (RTHS) is a problem formulation in single-agent search where the agent is required to iteratively plan and execute partial paths (since it does not have enough time to plan a full path). The key innovation in single-agent RTHS is that the agent updates (increases) the heuristic values of visited states. RTHS algorithms prove how this update prevents deadlock/livelock and enables completeness \cite{korf1990_lrta,koenig2006_rtaa,rivera2013weighted_real_time_search}.

\begin{figure}[t!]
    \centering
    \includegraphics[width=0.95\linewidth]{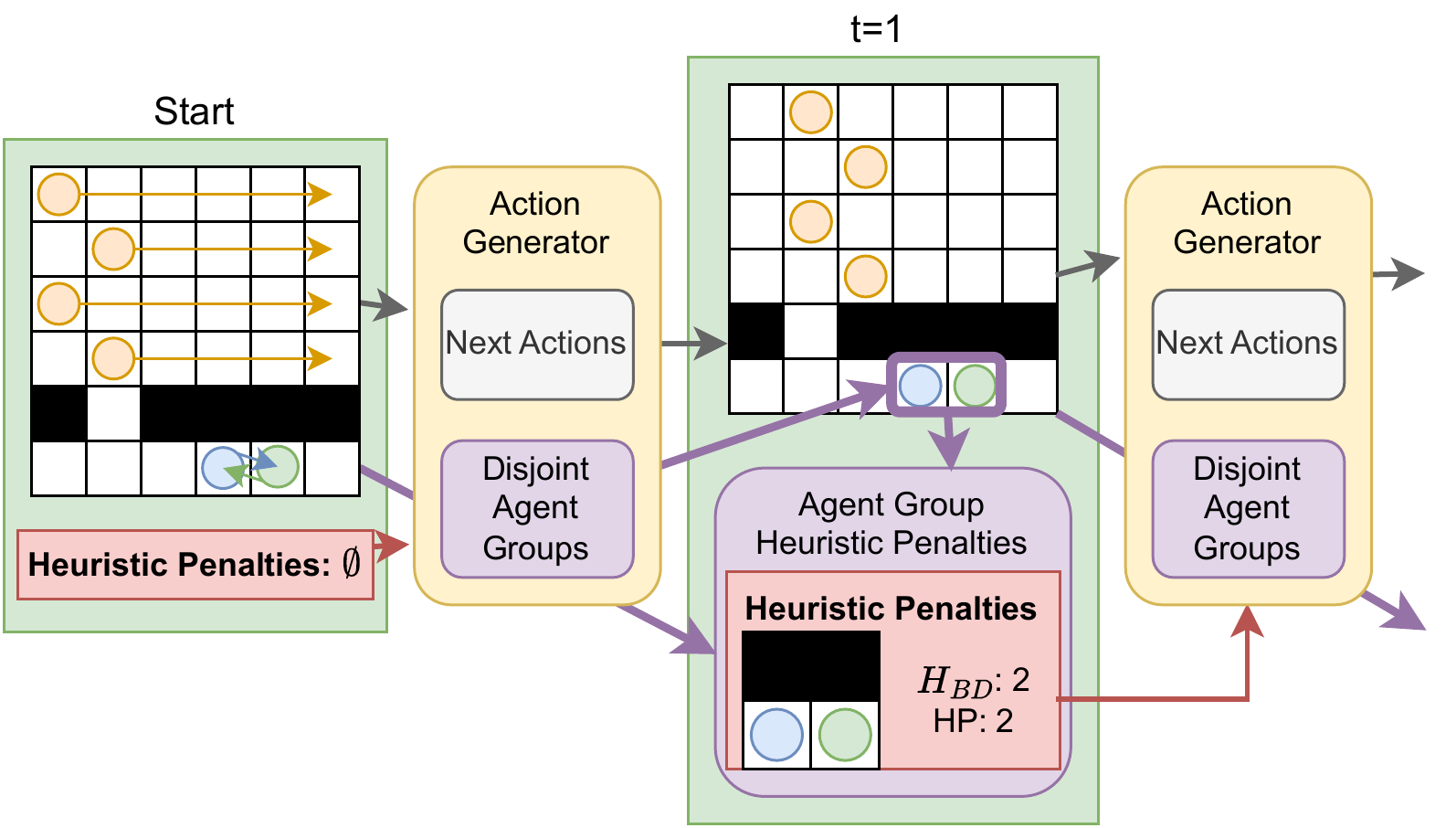}
    \vspace{-0.5em}
    \caption{The iterative planning and execution loop of WinC-MAPF with heuristic penalties and disjoint agent groups, figure taken from \citet{veerapaneni2024winc_mapf}.}
    \label{fig:winc-mapf-iterative}
    \vspace{-0.5em}
\end{figure}

\subsection{Windowed Complete MAPF} \label{sec:winf-mapf-background}
Our work directly extends the Windowed Complete MAPF (WinC-MAPF) framework. Due to the several moving parts of WinC-MAPF and the page limit, we provide a reduced summary of the work. Readers unfamiliar with WinC-MAPF are strongly encouraged to read the original paper, which has in-depth and robust explanations and examples.

The WinC-MAPF framework leverages ideas from real-time heuristic search to enable completeness for optimal windowed MAPF planners, which they term ``Action Generator" (AG).
Specifically, given an initial configuration $\comp^0$, they call the AG which returns the configuration $\comp^{W*}$ at the end of a partial solution $\comp^{0:W}$ that optimally minimizes $\sum_{i=1}^N ( \sum_{t=0}^{W-1} c(\comp_i^t,\comp_i^{t+1}) + h^{BD}_i(\comp_i^W))$, i.e., $\comp^{W*}$ is the best windowed configuration. Fig. \ref{fig:winc-mapf-iterative} depicts an overview of the WinC-MAPF framework.

\paragraph{Heuristic Update:} Given the optimal windowed solution $\comp^{0:W}$, they update the heuristic value of configuration $\comp^0$
\begin{equation} \label{eq:update}
h(\comp^0) \gets U(\comp^0, \comp^{W}) := \max(h(\comp^0), c(\comp^0, \comp^{W} )+ h(\comp^{W})),
\end{equation}%
where the heuristic $h(\comp)$ for all configurations are initially set to $\sum_{i=1}^N h_i^{BD}(\comp_i)$ and gradually increases as agents visit configurations.
For notional convenience, they denote $h(\comp) = h^{BD}(\comp) + h_p(\comp)$, where a ``heuristic penalty" term $h_p(\comp) > 0$ denotes an increase in heuristic over the default backward Dijkstra's value due to the update function $U$.


\paragraph{Disjoint Agent Groups:} The main drawback with this approach is that the number of configurations $\comp$ grows exponentially with the number of agents. Thus, escaping a local minima (e.g., deadlock/livelock in our case) can require updating the heuristic of an impractical number of configurations.
Thus, WinC-MAPF introduces the idea of disjoint agent groups, which essentially decomposes a windowed MAPF problem into sets of independent agent groups. Instead of updating the heuristic of the entire configuration $\comp$, they update the heuristics for $\comp_{Gr_i}$ where $\comp_{Gr_i}$ is the configuration of a group of agents $Gr_i$. In Fig. \ref{fig:winc-mapf-iterative}, this corresponds to detecting and updating the heuristic for the conflicting blue and green agents' grouped configuration as opposed to the entire joint configuration.

\begin{definition}[Disjoint Agent Groups]
    Given a configuration transition $\comp^0 \rightarrow \comp^W$, and set of disjoint agent groups $\{Gr_i\}$, each disjoint agent group $Gr_i$ satisfies the property that for each agent $j$ with transition $\comp^0_j \rightarrow \comp^W_j$ in $Gr_i$, there cannot exist another agent in a different group $Gr_k$ that blocked agent $j$ from picking a better path.
\end{definition}

Conceptually, this means that each group's decision to make $\comp^0_{Gr_i} \rightarrow \comp^W_{Gr_i}$ is independent of agents in other groups. Given a transition from $\comp^0 \rightarrow \comp^W$ and a set of disjoint agent groups $\{Gr_i\}$, we apply the update Eq. \ref{eq:update} to each group transition, i.e., $h(\comp^0_{Gr_i}) \gets U(\comp^0_{Gr_i},\comp^W_{Gr_i})$. This significantly improves efficiency by increasing the heuristic value of multiple configurations/groups at once.  
Given these group heuristic values, we compute the heuristic of a joint configuration $h(\comp)$ via the sum of each group's heuristic, i.e. $h(\comp) = \sum_{i} h(\comp_{Gr_i}) = h^{BD}(\comp) + \sum_i h_p(\comp_{Gr_i})$ for disjoint agent groups $\{Gr_i\}$ at the configuration.

A core component of WinC-MAPF is detecting / computing disjoint agent groups. Instead of computing disjoint groups where agents are not interacting, it is easier to determine coupled agent groups where they could be interacting.

\begin{definition}[Coupled Agents] \label{def:coupled-agents}
    Given a configuration transition $\comp \rightarrow \comp^W$, an agent $i$ is coupled with $j$ if $j$ prevents $i$ from choosing a better path or vice-versa.
\end{definition}

Note that coupled agents must be in the same disjoint agent group $Gr$. Importantly, disjoint agent groups do not need to solely consist of coupled agents, i.e., it is allowed for extra non-coupled agents (e.g., agents independent of all others) to be in a disjoint group.

\paragraph{Action Generator Requirements:} 
Combining everything together, the WinC-MAPF framework proves completeness with an optimal windowed MAPF Action Generator which (1) finds $\argmin_{\comp^W} c(\comp^0, \comp^W) + h(\comp^W)$, and (2) computes disjoint agent groups for $\comp^0 \rightarrow \comp^W$. The WinC-MAPF framework iteratively calls the AG to obtain a windowed plan, updates the heuristic of each disjoint agent group's configuration, executes the plan, and repeats (Fig. \ref{fig:winc-mapf-iterative}).

\paragraph{Single-Step CBS:} The WinC-MAPF paper also introduces Single-Step CBS (SS-CBS), which is a single-step optimal windowed Action Generator that is able to incorporate heuristic penalties and compute disjoint agent groups.
SS-CBS computes coupled agents (which form disjoint agent groups) by grouping together agents that shared a conflict/constraint during the process of finding the windowed solution.
However, incorporating heuristic penalties is trickier as \citet{veerapaneni2024winc_mapf} proves that they cannot be naively handled in CBS. 
Conceptually, SS-CBS resolves the problem by delaying incurring a heuristic penalty via a ``heuristic" conflict and constraint in the high-level search. A heuristic conflict occurs when agents match a group configuration with a non-zero $h_p(\comp_{Gr})$. SS-CBS then allows agents to try to avoid this heuristic penalty by replanning and avoiding $\comp_{Gr}$, or forces agents to stay at $\comp_{Gr}$ and incur the penalty.



The main restriction with WinC-MAPF is the requirement for an optimal windowed planner. In particular, although SS-CBS has a better success rate than naive windowed CBS, SS-CBS starts timing out when dealing with tougher congestion.


\paragraph{Proof of Completeness:} \label{sec:winc-proof-of-completeness}
The key to proving completeness in WinC-MAPF is showing that $h(\comp_{Gr_i}) \leq h^*(\comp_{Gr_i})$ holds even when updating the heuristic values, i.e., the heuristic of each disjoint agent group configuration $\comp_{Gr_i}$ is always admissible.
If that holds, then the heuristic $h(\comp)$ satisfies
\begin{equation} \label{eq:sum-of-groups-admissible}
    h(\comp):= \sum_{i} h(\comp_{Gr_i}) \leq \sum_{i} h^*(\comp_{Gr_i}) \leq h^*(\comp)
\end{equation}
for a set of disjoint agent groups $Gr_i$ that cover all agents.

Given this fact, proving completeness follows standard single-agent real-time heuristic search logic. The main intuition is that if the agents never reach the goal, they must be infinitely stuck cycling through a finite set of states. Across every cycle, we can show that at least one heuristic value must increase. This means that the heuristic values in the cycle are getting arbitrarily high, which contradicts heuristic values being admissible. Thus, the agents cannot be stuck in a loop and must eventually read the goal. Readers are encouraged to read \citet{veerapaneni2024winc_mapf} for more details about WinC-MAPF in general, and specifically their Appendix B.1 for a formal proof.

\section{Dynamic Agent Grouping ECBS}

The WinC-MAPF framework requires an \emph{optimal} Action Generator, which severely limits the range of solvers that can be used. This section first describes why naively trying to use a bounded suboptimal solver like ECBS will not satisfy the proof for completeness due to individual group's heuristics violating $w$-suboptimality (for the duration of this paper, $w$ denotes a suboptimality weight while $W$ denotes the planning window size). Then we describe our solution, Dynamic Agent Grouping ECBS (DAG-ECBS), which maintains that each group's heuristic is $w$-bounded, and show how it is complete within the WinC-MAPF framework.


\paragraph{Issue with Naïvely Using Bounded-Suboptimal Search: }
Conceptually, the WinC-MAPF proof in Section \ref{sec:winc-proof-of-completeness} relies on the fact that the heuristic $h(\comp)$ remains bounded throughout the planning and update procedure (i.e., Eq. \ref{eq:sum-of-groups-admissible}). This is done by showing that $h(\comp_{Gr}) \leq h^*(\comp_{Gr})$. If we want to use a bounded suboptimal search, we can extend/re-use this proof if we can show that $h(\comp) \leq w \cdot h^*(\comp)$ (i.e., is $w$-admissible), which we can do if we show that each group's heuristic $w$-admissible, i.e. $h(\comp_{Gr}) < w \cdot h^*(\comp_{Gr})$. Unfortunately, we show here that using a bounded-suboptimal solver such as ECBS does \textit{not} maintain this.

Applying a bounded-suboptimal solver such as ECBS with a single factor~$w$ over all agents to the WinC-MAPF framework ensures only that the total cost is within $w$ of optimal; individual groups can still be worse than $w$ times their own optimum.  
For example, if two groups have optimal costs $10$ and $40$, a global $w\!=\!2$ bound accepts any plan costing up to $100$.  A solution costing $30$ and $45$ for the two groups (total $75$) is globally acceptable, yet the first group violates its local $w$-bound of $20$, breaking the completeness proof when heuristics are updated.

This unfortunately means that we cannot use the existing proof of completeness from WinC-MAPF. Further attempts on our end for alternative proofs of completeness were unfruitful (and likely why the original WinC-MAPF paper focused only on optimal action generators).


\begin{algorithm}[t]
\textbf{Parameters}: Current configuration $\comp^0$, heuristic penalties $HPs$ \\ \noindent
\textbf{Output}: $w$-bounded suboptimal Windowed solution $\comp^{1:W}$, disjoint agent groups disjointGroups
\caption{Dynamic Agent Grouping ECBS}
\label{alg:dag-ecbs}
\begin{algorithmic}[1]
\State $\comp^{1:W} \gets \emptyset$ \Comment{Initialize empty windowed path}
\State activeGroups $\gets \{\{a\}\mid a\in\mathcal{A}\}$ \Comment{Singleton groups}
\State disjointGroups $\gets \{\}$
\While{\textit{activeGroups}$\neq\emptyset$}
    \State $Gr\gets$ \textit{activeGroups.pop()} \Comment{Current group to plan}
    \State $\hat\comp^{1:W}_{Gr} \gets$ \Call{Group-ECBS}{$Gr,\comp^{1:W}$, $HPs$}
    \State $Gr_{int} \gets \{Gr' | \comp^{1:W}_{Gr'} \text{~conflicts with~} \hat\comp^{1:W}_{Gr}\}$
    \If{$\hat\comp^{1:W}_{Gr} \neq \emptyset$ \textbf{and} $A_{int}=\emptyset$} \Comment{Successfully planned and not coupled with other agents}
        \State $\comp^{1:W}_{Gr} \gets \hat\comp^{1:W}_{Gr}$ \Comment{Update paths}
        \State disjointGroups.insert($Gr$) \Comment{Tentatively a disjointGroup}
    \ElsIf{$A_{int} \neq \emptyset$} \Comment{Coupled with other agents}
        \State merge $Gr$ with groups $Gr' \in Gr_{int}$ and re-insert into \textit{activeGroups}. Also remove $Gr'$ from activeGroups and disjointGroups if inserted.
    \Else \Comment{Failed to find a solution}
        \State \Return \textsc{Failure}
    \EndIf
\EndWhile
\State \Return $\comp^{1:W}$, disjointGroups
\end{algorithmic}
\end{algorithm}

\subsection{Dynamic Agent Grouping ECBS}
Dynamic Agent Grouping ECBS (DAG-ECBS) is a bounded suboptimal search that computes disjoint agent groups and returns a $w$-bounded suboptimal solution per agent group. DAG-ECBS is able to do this by planning agent groups \textit{independently} instead of planning across all groups at once. 
The idea of planning groups of agents is shared with Independence Detection \cite{standley2010operater_decomposition}, and in the context of CBS is similar to Meta-Agent CBS (MA-CBS)~\cite{sharon2015meta_agent_cbs} and Nested ECBS (NECBS)~\cite{chan2022nested_ecbs} which plan groups of colliding agents.
The purpose of grouping agents in MA-CBS and NECBS is to speed up the search. However, in our case we group agents as we need to compute a $w$-bounded suboptimal solution for each disjoint agent group.

Conceptually, DAG-ECBS essentially runs ECBS on individual agent groups (starting with each agent in its own group) and iteratively merges and replans colliding groups of agents. 
Algorithm \ref{alg:dag-ecbs} describes DAG-ECBS's main procedure of planning groups of agents in more detail. We maintain a queue of active groups ``activeGroups" which is initialized with singleton groups (i.e., we assign each agent its own individual group initially). We then proceed to plan each group $Gr$ in activeGroups (lines 4-6) by calling the Group-ECBS function. Group-ECBS finds a bounded suboptimal windowed solution for the specific agents in $G$.

The group's windowed plan is compared against other planned groups' windowed plans, and groups that \textit{conflict} are collected (line 7). Importantly, conflicts include heuristic conflicts in addition to standard vertex and edge conflicts. These groups are coupled with $Gr$. If there are no such groups, then $Gr$ is currently disjoint and the windowed plan is recorded (lines 8-10). If there are coupled groups, then the groups are merged with $Gr$, and then the new group is inserted back into the queue for later planning (lines 11-12).

This process means that DAG-ECBS starts with singleton groups but iteratively merges and replans groups that interact with each other. The final result is a set of disjoint agent groups and their collision-free paths, with the guarantee that each group's solution is $w$-bounded suboptimal.

\subsection{Group-ECBS} \label{sec:group-ecbs}
Group-ECBS is nearly identical to calling regular windowed ECBS on a group of agents, except for the following changes. First, and identical to Single-Step CBS in WinC-MAPF, Group-ECBS incorporates heuristic penalties and constraints to return bounded suboptimal solutions that incorporate changes in heuristics. Second, and unique to Group-ECBS, we modify the high-level and low-level focal search as we maintain that $h(\comp)$ is $w$-admissible. This modification is required for the proof of completeness.
%

Conceptually (but notationally imprecise), focal search \cite{PearTPAMIl1982FocalSearch} maintains $w$-bounded suboptimality by maintaining a lower bound $c + h$ (given an admissible $h$) and multiplying the quantity by $w$ for the focal threshold. We however maintain that our $h$ is $w$-admissible, so we instead keep track of and base our focal threshold on $w \cdot c + h$ to maintain that our solution is $w$-bounded suboptimal. 

We now proceed to define the focal search mathematically. For notational convenience, $n$ denotes a Constraint Tree (CT) node where conditioning on $n$ (which is denoted as $|n$) represents that $n$'s constraints must be satisfied.
Additionally since we want $h(\comp)$ to be $w$-admissible, we initialize $h(\comp)$ to $w \cdot h^{BD}(\comp)$. We thus redefine the heuristic penalty $h_p(\comp)$ to be the added residual, i.e. $h(\comp) = h_p(\comp) + w\cdot h^{BD}(\comp)$.

\paragraph{Low-Level Search: }
ECBS uses a focal search consisting of an ``Anchor" queue and a ``Focal" queue in the low-level search for planning single agent paths.
The priority function used in Anchor is the standard A* priority (Eq. \ref{eq:low-level-ecbs-anchor}). Eq. \ref{eq:low-level-ecbs-focal} describes which states $v$ are allowed in the Focal queue (e.g., which states agent $i$ can go to while maintaining the solution path is $w$-bounded suboptimal).
\begin{align}
    & \text{Anchor: } F_1(v|n) := c(\comp_i^0,v|n) + h_i^{BD}(v) \label{eq:low-level-ecbs-anchor} \\
    & \text{Focal: } \{v |c(\comp_i^0,v|n) \! + \! h_i^{BD}(v) \! \leq \! w \!\cdot \! \min_{v \in \text{Anchor}} \! F_1(v|n) \} \label{eq:low-level-ecbs-focal}
\end{align}
The low-level search in Group-ECBS in DAG-ECBS requires a subtle change to the Focal criteria. In particular, Eq. \ref{eq:low-level-dagecbs-focal}'s sole change is weighting $h_i^{BD}(v)$ by $w$.
\begin{align}
    & \text{Focal: } \{v | c(\comp_i^0,v|n) + w\!\cdot\! h_i^{BD}(v) \leq w\!\cdot\! \min_{v \in \text{Anchor}} \!
    F_1(v|n)\} \label{eq:low-level-dagecbs-focal}
\end{align}

\paragraph{High-Level Search: }
ECBS also uses a focal search for the high-level constraint search, with a similar Anchor queue and Focal queue. 
The priority function $F_2(n)$ used in Anchor is described in Eq. \ref{eq:high-level-ecbs-anchor} where $\min_{v \in \text{Anchor}_i} F_1(v|n)$ can be interpreted as the lower bound of a solution for agent $i$ that satisfies CT node $n$'s constraints. Eq. \ref{eq:high-level-ecbs-focal} describes which nodes are allowed in the Focal queue where $\comp^W$ is the configuration in $n'$.
\begin{align}
    & \text{Anchor: } F_2(n) := \sum_{i} \min_{v \in \text{Anchor}_i} F_1(v|n) \label{eq:high-level-ecbs-anchor} \\
    & \text{Focal: } \! \{n' | c(\comp^0,\comp^W) \! + \! h^{BD}(\comp^W) \! \leq \! w \! \cdot \! \min_{n \in \text{Anchor}} \! F_2(n)\} \label{eq:high-level-ecbs-focal}
\end{align}
Conceptually, Group-ECBS in DAG-ECBS requires also keeping track of heuristic constraints/penalties that increase the heuristic of the group. Eq. \ref{eq:high-level-dagecbs-anchor} includes $h_p(n)$ which shows how we incorporate heuristic penalties according to $n$'s constraints. Eq. \ref{eq:high-level-dagecbs-focal} shows how this value is used to determine the modified Focal threshold.
\begin{align}
    & \text{Anchor: } F_3(n) := h_p(n) + w \cdot \sum_i \min_{v \in \text{Anchor}_i} F_1(v|n) \label{eq:high-level-dagecbs-anchor} \\
    & \text{Focal: } \! \{ n' | c(\comp^0,\!\comp^W) \! + w \! \cdot \! h^{BD\!}(\comp^W) \! + \! h_p(n) \! \leq \!\!\!\!\! \min_{n \in \text{Anchor}} \!\!\!\! F_3(n) \} \label{eq:high-level-dagecbs-focal}
\end{align}

\subsection{DAG-ECBS within WinC-MAPF}
We directly use DAG-ECBS as an Action Generator within the WinC-MAPF framework of iterative planning, execution, and heuristic updates, and prove in the next section how we get completeness guarantees. 
Specifically, given a configuration $\comp^0$, we use DAG-ECBS to return the next windowed solution $\comp^{1:W}$ and disjoint agent groups. We then move the agents one step, update the heuristic of $\comp^0_{Gr}$ for each disjoint agent group $Gr$ via Eq. \ref{eq:update} (i.e. $h(\comp^0_{Gr}) \gets U(\comp^0_{Gr}, \comp^W_{Gr})$), and repeat this process.

We emphasize that DAG-ECBS dynamically groups agents at every planning and execution iteration. This means that agents can be not grouped together for the first few time steps if they are initially far apart, then grouped when they move near each other, and then afterwards be not grouped together if they move apart. 



One subtlety for readers familiar with single-agent RTHS is that we use suboptimality different than single-agent RTHS. Most single-agent RTHS methods use suboptimality to increase the rate of learning / heuristic updates. On the other hand, we use suboptimality to speed up the windowed planner, not to increase the heuristic update.


\subsection{Sketch of Proof of Completeness} \label{sec:proof-of-completeness}
The core of DAG-ECBS's modifications is to show that each group's heuristic remains $w$-admissible after the heuristic update equation is applied. As mentioned earlier, the intuition is that if $h$ is $w$-admissible, then $w \cdot c + h$ gives a $w$-admissible bound. The formal proof is given in the appendix.

With $w$-admissible heuristics and the finite state space, the standard RTHS / WinC-MAPF argument applies: the search cannot loop indefinitely (as this contradicts our proof of the heuristics always being $w$-admissible), so \textsc{DAG-ECBS} will not get stuck in a cycle and will eventually find a solution if one exists. We note that this proof enables designing other bounded-suboptimal windowed MAPF solvers / action generators and is more general than just DAG-ECBS.

One last subtlety is that since DAG-ECBS returns a windowed path, we also update the heuristic values of the intermediate configurations $\comp^{1:W-1}_{Gr}$ based on $\comp^W_{Gr}$ via a slightly modified Eq. \ref{eq:update}, which we detail in the appendix still retains the updated heuristic value is $w$-admissible.


\begin{figure*}[t!]  
    \centering
    \begin{subfigure}{0.95\textwidth}
        \centering
        \includegraphics[width=\textwidth]{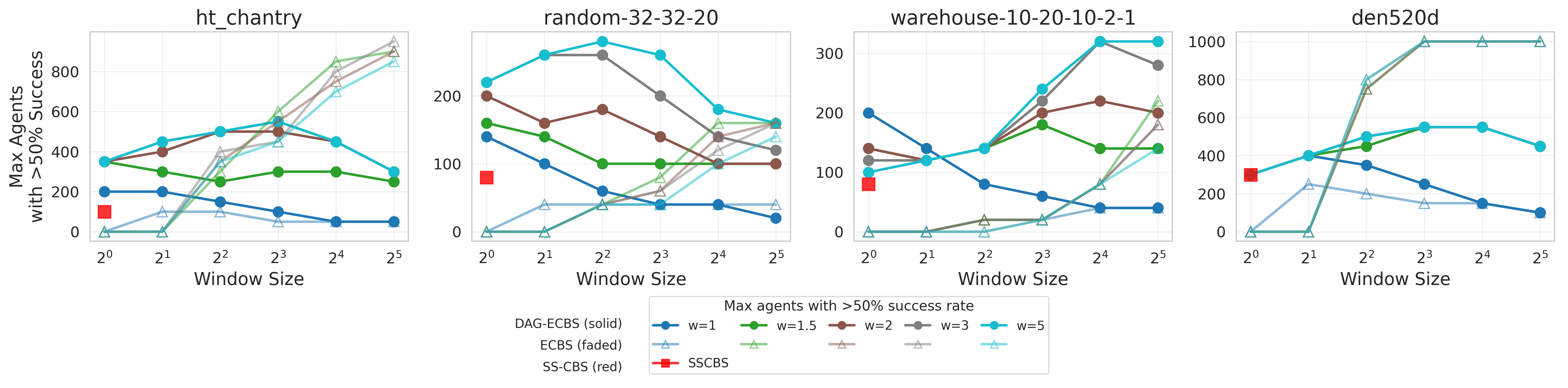}
        \caption{The maximum number of agents DAG-ECBS, windowed ECBS, and SS-CBS can scale to with a $>$ 50\% success rate.}
        \label{fig:success-rates-results}
    \end{subfigure}

    \vspace{1em}
    
    \begin{subfigure}{0.95\textwidth}  
        \centering
        \includegraphics[width=\textwidth]{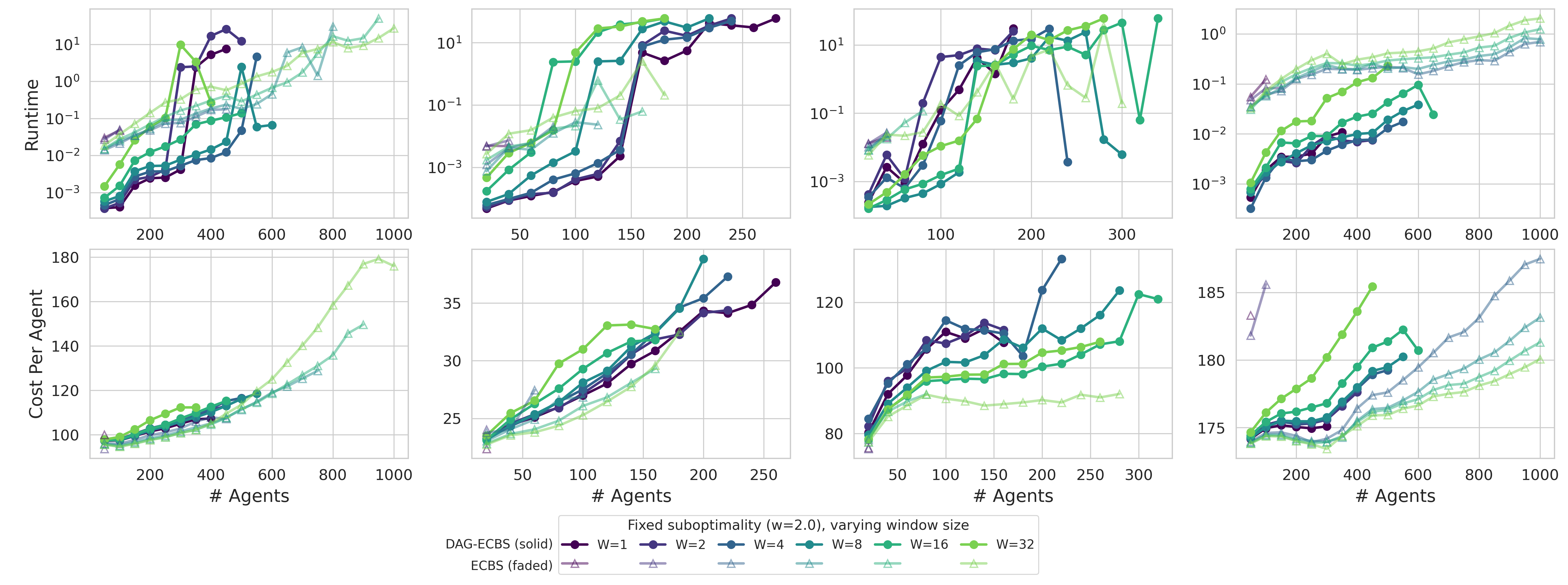}
        \caption{The effect of varying the window size $W$ in DAG-ECBS and windowed ECBS for suboptimality $w=2$.}
        \label{fig:varying-window-results}
    \end{subfigure}
    
    \vspace{1em}  

    \begin{subfigure}{0.95\textwidth}
        \centering
        \includegraphics[width=\textwidth]{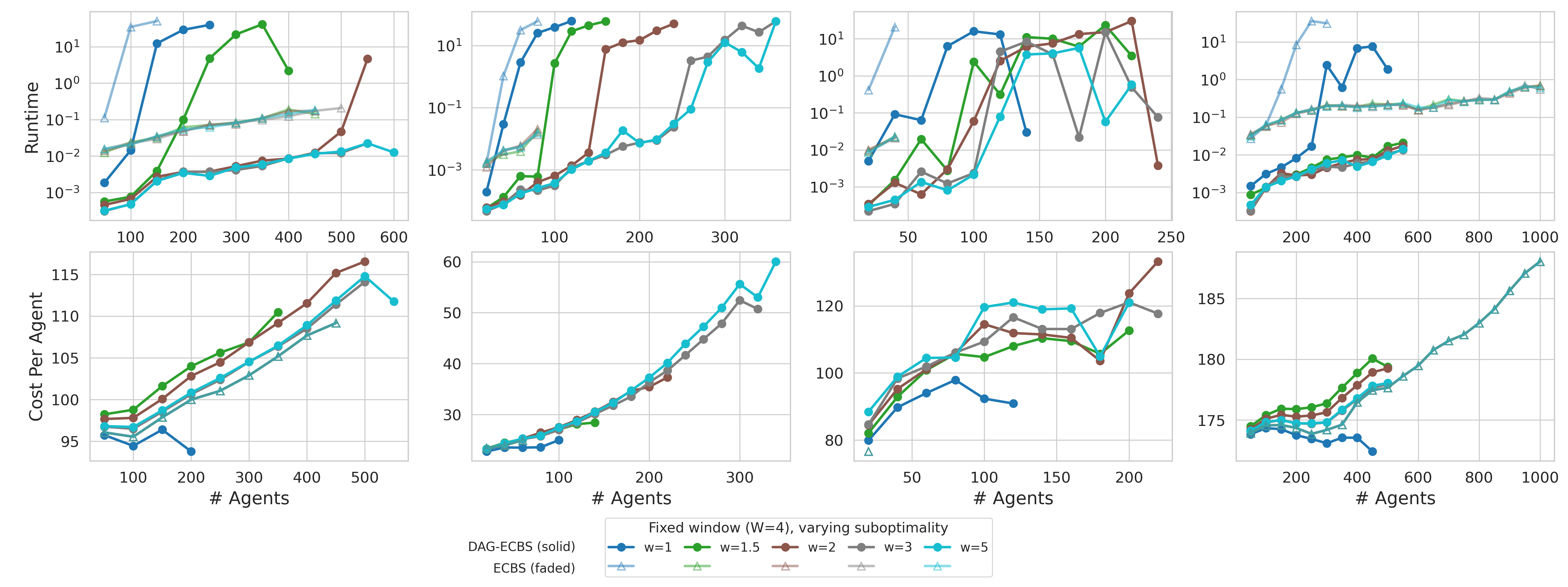}
        \caption{The effect of varying suboptimality $w$ in DAG-ECBS and windowed ECBS for fixed window $W=4$.}
        \label{fig:varying-suboptimality-results}
    \end{subfigure}
    \vspace{-0.5em}
    \caption{Evaluating DAG-ECBS and windowed ECBS across various suboptimalities ($w$), window sizes ($W$), and maps.}
    \label{fig:main-results}
    \vspace*{-1em}
\end{figure*}

\section{Experimental Results}
We evaluate DAG-ECBS and windowed ECBS on standard benchmark maps \cite{stern2019mapfbenchmark} with windows $W=\{1,2,4,8,16,32\}$ and suboptimalities $w=\{1, 1.5, 2, 3, 5\}$. We additionally include Single-Step CBS (SS-CBS) from the WinC-MAPF paper which is an optimal one-step solver with completeness guarantees. All methods interleave planning a windowed path, executing the first step, and repeating, with a 1 minute timeout summed across all planning iterations. Fig \ref{fig:main-results} shows results across different maps, parameters, and metrics.

We note we did not compare against state-of-the-art full-horizon methods (e.g., EECBS \cite{li2021eecbs} or LaCAM \cite{okumura2023lacam}) for two main reasons. First, the contribution of this paper is advancing windowed planning with completeness guarantees. Thus, the main comparison is with prior windowed planning with completeness (SS-CBS) or to test against the naive version without guarantees (windowed ECBS) to evaluate the empirical impact of theoretical guarantees. Second, most MAPF research has focused on full-horizon methods, and thus current full-horizon methods (e.g., EECBS, LaCAM) have better success rates than Fig \ref{fig:main-results}.

\paragraph{Scalability: }
Figure \ref{fig:success-rates-results} depicts the scalability of DAG-ECBS, ECBS, and SS-CBS as we vary the window size (x-axis) and suboptimality $w$ (lines). Note that SS-CBS is an optimal solver with window size one and thus only appears as one point (red). First, comparing DAG-ECBS (solid lines) to SS-CBS, we see that DAG-ECBS outperforms SS-CBS, highlighting the benefit of incorporating suboptimality into the search. Interestingly, DAG-ECBS with $w=1$ and $W=1$ outperforms SS-CBS (where they technically should find equivalent solutions). We surmise this is because planning agent groups helps performance (similar to NECBS \cite{chan2022nested_ecbs}). Second, observing ECBS (faded lines), we notice how performance is most tightly coupled with the window size; performance is poor for small window sizes and increases as the window size increases. The poor performance at small window sizes across all maps reveals how ECBS get stuck in deadlock/livelock in those situations. DAG-ECBS's superior performance show how incorporating heuristic updates reduces these issues.

Taking a step back, we see that DAG-ECBS consistently outperforms ECBS across all window sizes in maps random-32-32-20 and warehouse-10-20-10-2-1, but not for larger window sizes in ht\_chantry and den520d. This is because the random and warehouse map have more congestion with small agents that DAG-ECBS is able to overcome (and that ECBS gets stuck in). The other two maps are larger with less small agent congestion (i.e., congestion will consist of larger agents) where DAG-ECBS starts to timeout while ECBS is able to make progress.

\paragraph{Varying Window Size Analysis:}
Fig \ref{fig:varying-window-results} plots the maximum iteration runtime in seconds across each instance (failures included) and the average cost per agent (failures omitted). Each line denotes a different window size for DAG-ECBS and ECBS with suboptimality $w=2$ (SS-CBS is omitted as it only work with $w=1$). Comparing different colors lines, we see that increasing the window size increases the per-iteration runtime for both ECBS and DAG-ECBS, although it seems to affect the runtime of DAG-ECBS more. Combined with Fig \ref{fig:success-rates-results}, we see that DAG-ECBS fails due to runtime issues.

The effect on window size on cost varies across maps. In warehouse, larger window sizes reduce costs, while in the other three maps, they increase cost. We conceptualized two different explanations that appear to play out in different situations. First, a larger window size enables avoiding myopic decisions and therefore decreases overall solution length. Second, since we are employing suboptimality ($w=2$), a larger window size means that we compute a more suboptimal path and increases overall solution length. 



\paragraph{Varying Suboptimality Analysis:}
Fig \ref{fig:varying-suboptimality-results} keeps the window sized fixed ($W=4$) and varies the suboptimality (lines) of DAG-ECBS and ECBS.
First, observing the runtimes of DAG-ECBS, we see that increasing the suboptimality consistently reduces DAG-ECBS's per-iteration runtime. This boost in performance, especially in comparison to an optimal solver ($w=1$), highlights the importance of being able to use a bounded suboptimal solver for larger window sizes. 

ECBS seems to have a slightly runtime story. Like DAG-ECBS, ECBS with $w=1$ (blue triangle line) times out fast. However, unlike DAG-ECBS, ECBS with $w > 1$ all have similar runtimes (as opposed to DAG-ECBS where different $w>1$ have differing runtimes). We hypothesize this is due to the low-level focal queue. In particular, ECBS and DAG-ECBS have the same focal threshold value, but DAG-ECBS includes a $w \cdot h^{BD}(\comp_i)$ term on the left hand side of the inequality (Eq. \ref{eq:low-level-dagecbs-focal}) while ECBS does not (Eq. \ref{eq:low-level-ecbs-focal}). Since windowed ECBS is planning a partial path, this means many nodes can satisfy the focal threshold. On the flip end, there are less nodes in DAG-ECBS's focal queue. 
Analyzing the solution cost seems to verify this hypothesis. For DAG-ECBS, we see that increasing the suboptimality consistently increases the solution cost on all 4 maps. However for ECBS with $w>1$, all their solution costs seem to be similar.

\section{Conclusion}
The Windowed Complete MAPF (WinC-MAPF) framework
shows how to use windowed planning while retaining completeness guarantees. However, WinC-MAPF required an optimal windowed solver, restricting the types of useable solvers and limiting the scalability of the framework.

Our main contribution is extending the framework by showing how we can use a bounded suboptimal windowed solver, Dynamic Agent Grouping ECBS (DAG-ECBS), that maintains solution guarantees in the WinC-MAPF framework. The key idea in DAG-ECBS is to maintain that each groups heuristic is $w$-admissible by planning agent groups independently, with agent groups getting merged and replanned if they conflict. 

We highlight that the idea of grouping agents via conflicts and planning agent groups is broadly applicable to existing suboptimal MAPF algorithms. In fact, Alg \ref{alg:dag-ecbs} can be viewed as a template by replacing Group-ECBS with other algorithms in Line 6. Thus, DAG-ECBS serves as a template for converting suboptimal MAPF algorithms into Action Generators in the WinC-MAPF framework.
We empirically show that DAG-ECBS improves the scalability of WinC-MAPF compared to the single-step optimal equivalent. We additionally show that DAG-ECBS consistently outperforms windowed ECBS (without completeness guarantees) for small window sizes or in congested maps. 

We see two main avenues for future work. First,  we imagine adapting other more advanced bounded suboptimal MAPF algorithms like EECBS \cite{li2021eecbs} to follow our technique / template (e.g., creating a DAG-EECBS). Second, our dynamic grouping is reminiscent of decentralized multi-agent algorithms where agents near each other communicate (e.g., are grouped) while agents far apart do not. Thus, even though we implement DAG-ECBS as a centralized solver, we imagine that a decentralized version of DAG-ECBS is feasible.

\bibliography{aaai25}

\clearpage

\appendix

\setcounter{figure}{0}
\renewcommand{\thefigure}{A\arabic{figure}}
\setcounter{table}{0}
\renewcommand{\thetable}{A\arabic{table}}

\section{Quick Summary}
\subsubsection{Recommended background readings:} Readers new to MAPF, CBS, or ECBS are recommended to read ECBS \cite{barer2014suboptimalecbs}.
Readers new to single-agent Real-Time Heuristic Search are recommended to read LRTA* \cite{korf1990_lrta}. Readers new to Windowed Complete MAPF (WinC-MAPF) are strongly recommended to read \citet{veerapaneni2024winc_mapf}.

\subsubsection{Motivation in respect to prior work:} The majority of MAPF methods find collision-free paths to goals. However, this can be impractical in real world use cases with limited planning time and long solution paths. Instead, a more common real world scheme is to utilize windowed planning where the planner only reasons about a fixed $W$ timestep window and finds collision-free partial paths.
A key issue with windowed approaches is their myopic planning can result in deadlock or livelock. More broadly, until recently, all prior windowed MAPF solvers lacked provable completeness guarantees. 

Recently, the WinC-MAPF framework \cite{veerapaneni2024winc_mapf} showed how to obtain completeness even with windowed planning. They do this using heuristic updates and disjoint agent groups. However, WinC-MAPF only proved completeness for optimal windowed MAPF solvers. They also show Single-Step CBS, an instantiation of their framework using CBS. Our goal is to extend the ideas of Single-Step CBS and the WinC-MAPF framework to bounded suboptimal solvers.

\subsection{Intended Takeaways}
Our main contribution is to extend the WinC-MAPF framework by showing how we can use a bounded suboptimal windowed solver while maintaining solution guarantees. Our windowed solver, Dynamic Agent Grouping-ECBS (DAG-ECBS) dynamically groups agents and finds a $w$-bounded suboptimal solution per agent group by calling a modified ECBS on each agent group. Additionally, DAG-ECBS incorporates heuristic updates and returns disjoint agent groups. We prove how DAG-ECBS is complete in the WinC-MAPF framework.


1. DAG-ECBS plans agent groups independently and merges and replans agent groups that conflict. Unlike prior works (e.g., MA-CBS \cite{sharon2015meta_agent_cbs}, Nested ECBS\cite{chan2022nested_ecbs}) which do this to gain speed-ups, we do this to ensure that compute disjoint agent groups to ensure that each one has a bounded suboptimal solution.

DAG-ECBS also requires some subtle modifications to the priority function and focal bound to make sure that the heuristic values stay $w$-admissible after applying the heuristic update function.

2. We prove how using DAG-ECBS within the WinC-MAPF framework is complete. This proof more broadly shows how future bounded suboptimal MAPF methods can use the WinC-MAPF framework and have completeness guarantees.

3. Our results show that varying the window size and suboptimality has different effects based on the specific map. Overall, we see that for tough maps or small window sizes, DAG-ECBS has superior success rate and scalability over windowed ECBS without completeness guarantees.

\section{Proof of Completeness} \label{sec:appendix-proof-of-completeness}
We formally prove how the DAG-ECBS maintains completeness in the WinC-MAPF framework. The core component is proving that the heuristic values of each group is $w$-admissible even with heuristic updates. 

\begin{theorem}\label{thm:w-admissible}
For every disjoint agent group, the heuristic maintained by \textsc{DAG-ECBS} is always \emph{$w$-admissible}, i.e., $h(\comp_{Gr}) \leq w \cdot h^*(\comp_{Gr})$.
\end{theorem}

\begin{proof}
\textbf{Base case.} All possible group heuristics $h(\comp_{Gr})$ are $w$-admissible as they are initially set to $w \cdot  h^{BD}(\comp_{Gr}) = w \cdot \sum_{i \in Gr}h^{BD}_i(\comp_i)$ (Section \ref{sec:group-ecbs}). Note that $h_i^{BD}(\comp_i)$ is admissible as it is the individual agent's cost to reach the goal ignoring other agents.

\textbf{Inductive step.} We assume $h(\comp_{Gr}) \leq w \cdot h^*(\comp_{Gr})$ for all possible group configurations before planning. Then, DAG-ECBS is called and returns a set of disjoint agent groups and windowed paths. 
From the inductive hypothesis, we assume that existing $h$ are admissible for all group configurations. Each agent group's windowed path $\comp^W_{Gr}$ is the result of calling Group-ECBS where the solution node $n^{sol}$ contained the windowed solution. We show the following holds: 
\begin{align} \label{eq:inductive-step}
    & c(\comp^0_{Gr},\comp^W_{Gr}) + h(\comp^W_{Gr}) = \nonumber \\
    &c(\comp^0_{Gr},\comp^W_{Gr}) + w\cdot h^{BD}(\comp^W_{Gr}) + h_p(n^{sol}) \leq \nonumber \\
    & \min_{n \in \text{Anchor}} F_3(n) \! =  \!\min_{n \in \text{Anchor}} \! h_p(n) \!+\! w\!\cdot\!\sum_{i \in Gr} \! \min_{v \in \text{Anchor}_i} \!\!\! F_1(v|n) = \nonumber \\
    & \min_{n \in \text{Anchor}} h_p(n) + w \cdot \sum_{i \in Gr} F_1(\comp_i'|n) = \nonumber \\
    & \min_{n \in \text{Anchor}} h_p(\comp_{Gr}') + \sum_{i \in Gr} w \cdot c^*(\comp_i^0,\comp_i'|n) \! + \!w\cdot h^{BD}(\comp_i') = \nonumber \\
    & \min_{n \in \text{Anchor}} w \cdot c^*(\comp_{Gr}^0,\comp_{Gr}'|n) + h(\comp_{Gr}') \leq \nonumber \\ 
    & \min_{n \in \text{Anchor}} w \cdot c^*(\comp_{Gr}^0,\comp_{Gr}'|n) + w \cdot h^*(\comp_{Gr}') \leq w \cdot h^*(\comp_{Gr}^0)
\end{align}
where $\comp_i' \gets \arg \min_{v \in \text{Anchor}_i} F_1(v|n)$.

The first to second row uses the definition of heuristic penalties (Section \ref{sec:winf-mapf-background}, \ref{sec:group-ecbs}) and the fact that a valid solution in Group-ECBS must resolve heuristic conflicts.
The second to third row uses the focal threshold in Group-ECBS (Eq. \ref{eq:high-level-dagecbs-focal}). The third row substitutes using Eq. \ref{eq:high-level-dagecbs-anchor}.
The third to fourth row substitutes the minimum with $\comp'_i$. The next row substitutes using Eq. \ref{eq:low-level-ecbs-anchor}. We use $c^*(...)$ as opposed to $c(...)$ as the minimum $v$ in the low-level anchor search is guaranteed to search over optimal paths. 
The fifth to sixth row equality uses the definition of heuristic penalty (Section \ref{sec:group-ecbs}). 
The second to last inequality holds by the inductive hypothesis, and the last inequality holds as the Anchor in (Group-)ECBS admits the optimal solution.

Importantly, note that the $c(\comp^0_{Gr},\comp^W_{Gr}) + h(\comp^W_{Gr})$ term in the first row corresponds to the update equation $U(\comp^0_{Gr}, \comp^W_{Gr})$. Thus, applying the update equation still results in the updated heuristic being $w$-admissible.
\end{proof}

Given this fact, we can then reuse the standard Real-Time Heuristic Search / WinC-MAPF argument as described in Section \ref{sec:proof-of-completeness}. 
To our knowledge, this is the first time a focal search has been used in RTHS, so this proof/methodology may be of relevance to single-agent RTHS researchers.

One last subtlety is that since DAG-ECBS returns a windowed path, we also updated the heuristic values of the intermediate configurations $\comp_{Gr}^t$ for $t \in [1,W-1]$ based on the path to $\comp^W_{Gr}$. Specifically, we update using $h(\comp^t_{Gr}) \gets \max(h(\comp^t_{Gr}),  U(\comp^0_{Gr},\comp^W_{Gr}) - w \cdot c(\comp^0_{Gr},C^t_{Gr}))$ which we can prove is $w$-admissible if $U(\comp^0_{Gr},\comp^W_{Gr})$ is $w$-admissible (and hence retains overall completeness guarantees) by rearranging terms. In the following proof we drop the group subscript for notational simplicity:
\begin{align}
    & U(\comp^0,\comp^W) \leq w \cdot h^*(\comp^0) \nonumber \\
    & U(\comp^0,\comp^W) \leq w\cdot [c^*(\comp^0,\comp^t) + h^*(\comp^t)] \nonumber \\
    & U(\comp^0,\comp^W) - w \cdot c^*(\comp^0,\comp^t) \leq w \cdot h^*(\comp^t) \nonumber \\
    & U(\comp^0,\comp^W) - w \cdot c(\comp^0,\comp^t) \leq w \cdot h^*(\comp^t) 
\end{align}
The first line is the result of Eq. \ref{eq:inductive-step}. The transition to the second line is via $h^*$ being consistent (an optimal heuristic is always consistent). The transition to the third line re-arranges terms. The final line uses the fact that $c^*(...) \leq c(...)$ by definition, so subtracting by a larger value will still retain the inequality.


\end{document}